\documentclass[12pt]{article}
\usepackage{amsmath}
\usepackage{graphicx,psfrag,epsf}
\usepackage{enumerate}
\usepackage[figuresright]{rotating}
\usepackage{natbib}
\usepackage{amssymb}
\usepackage{pifont}
\usepackage{multirow}
\usepackage{svg}
\usepackage{booktabs}
\usepackage{url} 
\usepackage{adjustbox}
\usepackage{booktabs}
\usepackage{threeparttable}
\usepackage{authblk}

\usepackage{amsmath}
\usepackage{graphicx,psfrag,epsf}
\usepackage{enumerate}
\usepackage[figuresright]{rotating}
\usepackage{natbib}
\usepackage{amssymb}
\usepackage{pifont}
\usepackage{multirow}
\usepackage{booktabs}
\usepackage{url} 
\usepackage{adjustbox}
\newcommand{\blind}{0}

\usepackage{amsmath}
\usepackage{tcolorbox}
\usepackage{subfig}
\usepackage{amssymb}
\usepackage{mathtools}
\usepackage{amsfonts}
\usepackage{multirow}
\usepackage{amsthm}
\usepackage{mathrsfs}
\usepackage{caption}
\usepackage{indentfirst}
\newtheorem{theorem}{Theorem}
\newtheorem{lemma}{Lemma}
\newtheorem{corollary}{Corollary}
\newtheorem{proposition}{Proposition}
\theoremstyle{definition}

\newtheorem{remark}{Remark}
\newtheorem{assumption}{Assumption}

\usepackage{fancybox,framed}

\usepackage[utf8]{inputenc}
\usepackage{graphicx}

\usepackage{float}
\begin{document}

\newcommand\x{\mathbf{x}}
\newcommand\uu{\mathbf{u}}
\newcommand\vv{\mathbf{v}}
\newcommand\V{\mathbf{V}}
\newcommand\y{\mathbf{y}}
\newcommand\Z{\mathbf{Z}}
\newcommand\Q{\mathbf{Q}}
\newcommand\PP{\mathbf{P}}
\newcommand\X{\mathbf{X}}
\newcommand\Y{\mathbf{Y}}
\newcommand\f{\mathbf{f}}
\newcommand\e{\mathbf{e}}

\newcommand\SSigma{\boldsymbol{\Sigma}}
\newcommand\halpha{\hat{\alpha}}
\newcommand\PPi{\boldsymbol{\Pi}}
\newcommand\bbeta{\boldsymbol{\beta}}
\newcommand\colm[1]{{\color{magenta}#1}}

\def\spacingset#1{\renewcommand{\baselinestretch}%
{#1}\small\normalsize} \spacingset{1}


\if0\blind
{
  \title{ \Large\textbf{The First-stage F Test with Many Weak Instruments\footnote{We are grateful to Anna Mikusheva, Stanislav Anatolyev, Federico Crudu and  participants of the SWETA2023 workshop, IAAE2023 and AMES2023 for helpful comments.} }}
  
\author{Zhenhong Huang}
\affil{\normalsize{\textit{Department of Statistics and Actuarial Science, The University of Hong Kong}} \protect \\ \small{\textit{(e-mail: zhhuang7@connect.hku.hk)}}}

\author{Chen Wang}
\affil{\normalsize{\textit{Department of Statistics and Actuarial Science, The University of Hong Kong}} \protect\\ \small{\textit{(e-mail: stacw@hku.hk)}}}

\author{Jianfeng Yao}
\affil{\normalsize{\textit{School of Data Science, Chinese University of Hong Kong (Shenzhen)} \protect \\ \small{\textit{(e-mail: jeffyao@cuhk.edu.cn)}}}
}

    \date{}
   \maketitle}
 \fi

\if1\blind
{
  \bigskip
  \bigskip
  \bigskip
  \begin{center}
    {\LARGE\bf Title}
\end{center}
  \medskip
} \fi

\bigskip
\begin{abstract}
\noindent A widely adopted approach for detecting weak instruments is to use the first-stage $F$ statistic. While this method was developed with a fixed number of instruments, its performance with many instruments remains insufficiently explored. We show that the first-stage $F$ test exhibits distorted sizes for detecting many weak instruments, regardless of the choice of pretested estimators or Wald tests. These distortions occur due to the inadequate approximation using classical noncentral Chi-squared distributions. As a byproduct of our main result, we present an alternative approach to pre-test many weak instruments with the corrected first-stage $F$ statistic. An empirical illustration with \cite{angrist1991does}'s returns to education data confirms its usefulness.
\end{abstract}

\noindent%
JEL Classiﬁcation numbers: C12, C26\\
Word Count: 5740
\vfill

\newpage
\spacingset{1.8} 
\section{Introduction}
\label{sec:intro}
The first-stage $F$ statistic introduced by \citet[hereafter referred to as SY2005]{stock2002testing} is commonly used to detect weak instruments in empirical research. Evidence of its popularity can be found in American Economic Review, where 15 of 17 papers published between 2014 and 2018 reported at least one first-stage $F$ statistic  \citep{andrews2019weak}.
However, this approach was originally developed for a fixed number of instrumental variables (IVs), and does not address the case of a large number of instruments, which is commonly encountered in practice \citep[see, e.g.,][]{angrist1991does, dobbie2018effects,bhuller2020incarceration}.

Several studies have pointed out limitations of applying SY2005's $F$ test with many instruments. For example, \cite{hansen2008estimation} demonstrated through empirical examples and simulations that a low $F$ statistic does not necessarily indicate weak instruments.  More recently, \citet[hereafter referred to as MS2022]{mikusheva2021inference} described that the classical $F$ test can mistakenly identify weak instruments mainly due to the insufficiency of the conventional measure for instrument strength, known as the concentration parameter. However, these studies only narratively discussed the unreliability of the $F$ test.  The theoretical basis for not recommending the $F$ test in practice has yet to be established.

In this paper, we study the asymptotic behavior of the first-stage $F$ statistic within the many-instrument framework, where the number of instruments and the sample size go to infinity simultaneously and proportionally. We show that the more appropriate distribution of the $F$ statistic shifts to the normal distribution, instead of the conventional noncentral Chi-squared distribution. The inadequacy of the noncentral Chi-squared distribution provides poor finite sample approximations to the $F$ statistic with many instruments, leading to size distortion of the classical $F$ test. These size distortions occur regardless of the pretested IV estimator or Wald test and become increasingly severe as the number of instruments approaches the sample size.

Our second goal is to correct SY2005's two-step procedure to enhance the usability of the $F$ test with many instruments. Apart from the inadequacy of the noncentral Chi-squared distribution, SY2005's two-step procedure suffers from the insufficiency of the concentration parameter when measuring instrument strength. In the case of many instruments,  \cite{chao2005consistent} and MS2022 show that the appropriate measure is the \textit{re-scaled} concentration parameter, which is the ratio of the concentration parameter over the \textit{square root} of
the number of instruments.  In our asymptotic result, the re-scaled concentration parameter appears in the centering term  of the $F$ statistic. Building on this, we propose a two-step procedure based on the $F$ statistic to detect many weak instruments that is analogous to that of MS2022. Our proposed statistic is directly derived from the classical $F$ statistic and follows the standard normal distribution, making it both conceptually familiar and straightforward to apply.

By identifying the deficiencies of the first-stage $F$ statistic with many instruments, this study contributes to the literature on discussing its limitations and implications for empirical analysis. In the case of a fixed number of instruments, \cite{lee2022valid} and \cite{keane2023instrument} focus on the performance of the IV t-test and show that using the rule-of-thumb $F > 10$ as a diagnostic cannot guarantee its well-controlled size and power. They further suggest that a higher threshold should be adopted in practice. Our study provides theoretical justification for the unreliability of the $F$ test to gauge instrument strength in many-instrument settings.

Additionally, this study contributes to  the literature on measuring the strength of many instruments. \cite{hahn2002new}
proposed a test to examine the adequacy of the standard asymptotic result in IV regression models. They argue that if the test rejects their null, then weakness in instruments may arise. However, \cite{lee2012hahn} proved that it is indeed a test for the exogeneity of the instruments. MS2022 and \cite{carrasco2022testing} considered heteroscedastic models and  proposed novel $F$-type tests for many weak instruments. Our study focuses on the original $F$ test statistic and makes corrections for the effects of many instruments.

The paper is organized as follows. In Section \ref{sec:model}, we introduce the model, followed by a discussion of the concentration parameter. In Section \ref{sec:pitfall}, we show the unreliability of the first-stage $F$ test by proving its size distortions through the noncentral Chi-squared approximation. In Section \ref{sec:corF}, we propose a two-step procedure using the first-stage $F$ statistic with many instruments.  Section \ref{sec:emp} presents an analysis of the returns to education data in \cite{angrist1991does}.  Section \ref{sec:con}  concludes with some further discussions.

\section{Model setup} \label{sec:model}
We consider the following model:
\begin{equation}
    y_i=Y_i\beta+u_i,
\end{equation}
\begin{equation} \label{1stage}
    Y_i=\pi'\Z_i+v_i,
\end{equation} for $i=1,\dots,n$, where $y_i$ is a scalar outcome, $Y_i$ is a scalar endogenous variable, $\Z_i$ is a $K_n\times 1$ vector of instrument variables. Errors $(u_i,v_i)$ have zero mean, covariance $\sigma_{vu}$ and variances $\sigma_{uu}^2$ and $\sigma_{vv}^2$, respectively.  We denote by $\y$, $\Y$, $\uu$, and $\vv$ the  $n\times 1$ vectors that collect the corresponding scalars, $\Z=(\Z_1',\dots,\Z_n')'$  the $n\times K_n$ matrix of observations on the $K_n$ instrumental variables. Moreover, $\PP_Z=\Z(\Z'\Z)^{-1}\Z'$, and $\mathbf{M}_Z=\mathbf{I}_n-\PP_Z$ are two projection matrices, and $\mathbf{D}_Z=\mathrm{diag}(P_{11},\dots,P_{nn})$ is the diagonal matrix containing the diagonal terms of $\PP_Z$. 

The behaviors of IV estimation and inference methods crucially depend on the magnitude of the concentration parameter, 
\begin{equation}
    \mu_n^2=\frac{\pi'\Z'\Z\pi}{\sigma_{vv}^2},
\end{equation} which characterizes the strength of instruments. When $K_n\equiv K$,  SY2005 demonstrated that a small value of $\mu_n^2$ indicates weak instruments. When $K_n \rightarrow \infty$, a more appropriate measure of the strength of instruments is   ${\mu_n^2}/{\sqrt{K_n}}$ that leverages the effect of many instruments.  \cite{chao2005consistent} showed that  the bias-corrected 2SLS (B2SLS) estimator \citep{nagar1959bias} estimator, the limited information maximum likelihood (LIML) estimator \citep{anderson1949estimation} and the jackknife instrumental variable estimator (JIVE) \citep{angrist1995split}  are consistent only when $\mu_n^2$ grows faster than $\sqrt{K_n}$. Wald-tests based on the above estimators therefore over-reject when  $\mu_n^2/ \sqrt{K_n}$ is bounded. Furthermore, MS2022 showed that there exists no consistent test for testing $\beta=\beta_0$ when $\mu_n^2/\sqrt{K_n}$ stays bounded; for this reason, they defined the instruments to be weak if $\mu_n^2/ \sqrt{K_n}$ stays bounded. 
We therefore focus on the measure $\mu_n^2/ \sqrt{K_n}$, which characterizes instrument strength within the many instruments framework.

\section{Size distortions of the classical $F$ test} \label{sec:pitfall}
In this section, we first review SY2005's influential $F$ test for detecting weak instruments, and show that it has distorted sizes when detecting many weak instruments\footnote{Stock-Yogo also showed that the $F$ test remains valid when $K_n^4/n \rightarrow 0$. However, this condition in fact requires very small $K_n$. For example, when the sample size is large enough to reach 10000, the number of instruments should be much smaller than 10 to satisfy the asymptotic scheme. Therefore, this setting cannot cover practical situations where $K_n$ is in hundreds.}. 

SY2005 defines instruments to be weak if the bias of IV estimators (e.g., the 2SLS-OLS relative bias) or rejection rate of IV-Wald tests (e.g., the 2SLS-Wald test) exceeds a predetermined tolerance level (e.g., 10\%).  They further showed that $\mu_n^2$ can fully determine both the level of the estimation bias and rejection rate. Therefore, in SY2005's first step, a theoretical value of $\mu_n^2=\mu_0^2$ that indicates weak instruments is obtained. In the second step, SY2005 proposed to use the first-stage $F$ statistic to test $H_0^{SY}:\mu_n^{2}\leq \mu_0^2$ and showed that: 
\begin{equation} \label{SY_F}
    F=\frac{\Y'\PP_{Z}\Y/K}{\Y'\mathbf{M}_{Z}\Y/(n-K)} \stackrel{d}{\rightarrow} \frac{\chi^2_{K}(\mu_0^2)}{K}
\end{equation}
when $K_n\equiv K$, where $\chi^2_{K}(\mu_0^2)$ denotes the non-central Chi-squared distribution with $K$ degrees of freedom and  noncentrality parameter $\mu_0^2$. To summarize, SY2005's two-step testing procedure for weak instruments is formulated as follows:
\begin{enumerate}
    \item Obtain $\mu_0^2$ by controlling the worst estimation bias of IV estimators or worst size distortions of IV-Wald tests.  
    \item Determine a critical value for $F$ by (\ref{SY_F}).
\end{enumerate}

We first examine the empirical sizes of the first-stage $F$ statistic using simulations. Let  $n=1000$ and  $K_n=5,300,500$ and 800. Consider  $\beta=1$, $\boldsymbol{Z}_i \stackrel{i.i.d.}{\sim} N_{K_n}(\mathbf{0},\mathbf{I}_{K_n})$, 
$\{v_i\}_{i=1}^n$ are i.i.d. normal with $\sigma_{vv}^2=1$, and $\mu^2_0=5$ and 500.
\begin{table}[]
    \centering
    \begin{tabular}{c|ccccc}
    \hline\hline
        $K_n=$ & 5 & 300 & 500 & 800 \\
        \hline
        $\mu_0^2=5$ & 5.2 & 8.45 & 12.6 &23.0 \\
        $\mu_0^2=500$ & 5.1 & 8.3 & 12.8 &22.8  \\
        $ \Phi\left(\sqrt{1-\frac{K_n}{n}}\Phi^{-1}(0.05)\right)$ & 5.1 & 8.4 & 12.2 &23.1\\
        \hline\hline
    \end{tabular}
    \caption{First two rows: empirical sizes of the conventional $F$ test. Third row: theoretical sizes predicted by Theorem \ref{size_F}.  Replication time is 2,000.  }
    \label{tab:size_conF}
\end{table}
The first two rows of Table  \ref{tab:size_conF} report that the conventional $F$ test has correct sizes  with a fixed number of instrument, but over-rejects $H_0^{SY}$ when the number of instruments becomes large, regardless of the magnitude of $\mu_0^2$.  Moreover, the over-rejection phenomenon gets increasingly severe when $K_n$ gets close to $n$. For example, when $\mu_0^2=5$ and $K_n$ increases from 500 to 800, the empirical sizes increases from 12.6\% to
23\%, which both far exceed the nominal level 5\%. 

It is natural to expect that the distribution in (\ref{SY_F}) can explain the size distortion phenomenon in Table  \ref{tab:size_conF} after letting $K \rightarrow \infty$. However, the expectation for this sequential limit scheme (SEQ-L: $n\rightarrow\infty$, followed  by $K_n\rightarrow\infty$)  turns out to be incorrect. Specifically, after renormalizing the noncentral Chi-squared distribution, the SEQ-L will provide the CLT: $ \sqrt{K_n}\left(F-1-{\mu_n^2}/{K_n }\right)\stackrel{d}{\rightarrow}N(0,2)$ that leads to the following result:
\begin{proposition} \label{theo:seql}
    Under the SEQ-L, we have
    \begin{equation*}
    \mathbb{P}\left( F > \frac{q_{\tau}^{\chi^2_{K_n}(\mu_0^2)}}{K_n}\right) \rightarrow \tau.
\end{equation*}
\end{proposition}
Proposition \ref{theo:seql} shows that the SEQ-L predicts the classical $F$ test to have correct sizes with many instruments. Therefore, it fails to characterize the size distortion phenomena observed in Table \ref{tab:size_conF}. Such inadequacy of the SEQ-L motivates us to study the asymptotic behaviour of the $F$ statistic under the simultaneous limit scheme (SIM-L), where $n$ and $K_n$ go to infinity simultaneously and proportionally.  The following assumptions are used in the sequel.

\begin{assumption}{(SIM-L)}\label{assum:1} As $n\rightarrow \infty$, $K_n/n \rightarrow\alpha \in (0,1)$.
\end{assumption} 

\begin{assumption}\label{assum:2} The first-stage errors $\{v_i\}_{i=1,\dots, n}$ are i.i.d. with finite fourth moment.
\end{assumption} 

Assumption \ref{assum:1}  is standard in the many IV literature which was initially introduced in \cite{bekker1994alternative}. 
Assumption \ref{assum:2} assumes the homoscedastic first-stage errors. Our results are established under homoscedasticity
as we focus on the behaviour of the original $F$ statistic, which was developed in such context.  Investigating the performance of the $F$ statistic under heteroscedastiticty is beyond the scope of this paper. We establish the limiting distribution of the first-stage $F$ statistic for a large $K_n$ in the following theorem.

\begin{theorem} 
 \label{theo:F} When $\mu_n^2/K_n \rightarrow 0$ and under Assumptions \ref{assum:1} and \ref{assum:2}, as $n\rightarrow\infty$,
\begin{equation} \label{asym:F}
    \sqrt{n}\left(F-1-\frac{\mu_n^2}{K_n }\right)\stackrel{d}{\rightarrow}N\left(0,\frac{(\omega-\alpha^2)\mathrm{E}(v_1^4)+(2\alpha-3\omega+\alpha^2)\sigma_{vv}^4}{\alpha^2(1-\alpha)^2\sigma_{vv}^4}\right),
\end{equation}
where $ \omega=\lim _{n \rightarrow \infty} \frac{1}{n}\sum_{i=1}^nP_{ii}^2$.
\end{theorem}
The condition  $\mu_n^2 =o(K_n)$ is relatively weak as it
covers both the weakly identified and strongly identified cases. This condition is also made in MS2022. The asymptotic normality of the $F$ statistic, as shown in (\ref{asym:F}), stands in stark contrast to the conventional Chi-squared distribution, which only holds for a fixed number of instruments.  Applying Theorem \ref{theo:F}, the following corollary confirmed the size distortions of the classical $F$ test observed in Table \ref{tab:size_conF}.
 \begin{corollary}\label{size_F} 
When $\mu_n^2/K_n\rightarrow 0$, under Assumptions \ref{assum:1} and \ref{assum:2}, as $n\rightarrow\infty$, 
\begin{equation*}
    \mathbb{P}\left( F > \frac{q_{\tau}^{\chi^2_{K_n}(\mu_0^2)}}{K_n}\right) \rightarrow \Phi\left(\sqrt{\frac{2\alpha(1-\alpha)^2\sigma_{vv}^4}{(\omega-\alpha^2)\mathrm{E}(v_1^4)+(2\alpha-3\omega+\alpha^2)\sigma_{vv}^4}}\Phi^{-1}(\tau)\right),
\end{equation*}
where $\tau$ is the significance level, $q_{\tau}^{\chi^2_{K_n}(\mu_0^2)}$  is the $(1-\tau)$-quantile of $\chi^2_{K_n}(\mu_0^2)$ distribution, and $\Phi^{-1}(\cdot)$ denotes the inverse cumulative distribution function of a standard normal random variable. Furthermore, suppose that $\omega=\alpha^2$ or $v_i$ have zero excess kurtosis, then
\begin{equation*}
    \mathbb{P}\left( F > \frac{q_{\tau}^{\chi^2_{K_n}(\mu_0^2)}}{K_n}\right) \rightarrow \Phi\left(\sqrt{1-\alpha}\Phi^{-1}(\tau)\right)>\tau.
\end{equation*}

 \end{corollary}
Corollary \ref{size_F} theoretically identifies the limitation of the first-stage $F$ test with many instruments due to the poor approximation using the noncentral Chi-squared distribution.  Particularly, when dealing with asymptotically balanced instruments ($\omega=\alpha^2$) or mesokurtic first-stage errors, the classical $F$ test would be oversized.  Moreover, the size distortions become more severe as $\alpha$ approaches 1. 

Corollary \ref{size_F} shows that our result under the SIM-L successfully recognizes the size distortion phenomena. The ratio $\alpha$ plays a crucial role  as it depicts the effect of the magnitude of $K_n$ that is invisible under the SEQ-L. This difference between the asymptotic behaviours of $F$ under the SIM-L and SEQ-L allows us to explain from a theoretical perspective the over-rejection phenomenon of the classical $F$ test when the number of instruments is relatively large. The last row in Table \ref{tab:size_conF} reports the predicted sizes from Theorem \ref{size_F}, which aligns perfectly with the empirical counterpart in the first two rows.

Corollary \ref{size_F} also serves as a warning to researchers using the classical $F$ test to detect many weak instruments of 
the size distortion problem, no matter which IV estimator or IV-Wald test is pre-tested. For example, relying on the popular rule-of-thumb that compares $F$ and the cutoff of 10 can still fail to control the rejection rate of B2SLS-Wald test within 10\%. Therefore, empirical researchers are warned not to use the classical $F$ test to detect many weak instruments.

 \begin{remark}
     From Theorem \ref{size_F}, the classical $F$ test will have asymptotically correct size if and only if 
     \begin{equation*} 
         (\omega-\alpha^2)\mathrm{E}(v_1^4)+(5\alpha^2-3\omega-2\alpha^3)\sigma_{vv}^4=0.
     \end{equation*}
     However, verifying this condition is challenging since the moments of errors are typically unknown. Even if this condition is satisfied,  the first step of SY2005's procedure is invalid within many-instrument setup, making the classical $F$ test remains deficient, see detailed discussions in Appendix. 
 \end{remark}

\section{The corrected $F$ test for many weak instruments} \label{sec:corF}
To enhance the usability of the first-stage $F$ statistic, we present a new two-step procedure for many weak instruments. In the first step, we consider controlling the worst rejection rate of the B2SLS-Wald test as B2SLS is consistent in the homoscedasticity setting when $\mu_n^2/\sqrt{K_n}\rightarrow \infty$.  In the second step, we propose a corrected  $F$ test to assess the reliability of the B2SLS-Wald test.

We re-consider the behaviour of the B2SLS-Wald test statistic in Section 3.4 of SY2005:
\begin{equation}
    W=\frac{n(\hat{\bbeta}_{B2SLS}-\beta_0)^2}{\hat{V}},
\end{equation}
where 
\begin{equation*}
    \hat{\bbeta}_{B2SLS}=\frac{\Y'\PP_b\y}{\Y'\PP_b\Y}, \; \hat{V}=\frac{n-K_n}{\Y'\PP_b\Y}\hat{\sigma}_{uu}^2+\frac{K_n}{n-K_n}\frac{\hat{\uu}'\mathbf{M}_Z\hat{\uu}\Y'\mathbf{M}_Z\Y+(\Y'\mathbf{M}_Z\hat{\uu})^2}{(\Y'\PP_b\Y)^2},
\end{equation*}
 with $\PP_b=\PP_Z-K_n/n\mathbf{I}_n$, $\hat{\uu}=\y-\Y\hat{\bbeta}_{B2SLS}$ and $\hat{\sigma}_{uu}^2$ being the B2SLS-residuals estimator. To test for $H_0:{\mu_n^2}/{\sqrt{K_n}}\leq C$, we propose a corrected $F$ test using statistic
 \begin{equation} \label{eq:ftest}
    F_{c}=\sqrt{\frac{K_n(n-K_n)}{2n}}\left[F-1-\frac{C}{\sqrt{K_n}}\right],
\end{equation}
  where $C$ is a constant obtained in our first step that is formulated later. We establish the behaviour of the B2SLS-Wald statistic as follows:
\begin{theorem} \label{theo:wald}
     Let Assumptions \ref{assum:1} and \ref{assum:2} hold. Assume that  $\mu_n^2/K_n \rightarrow 0$ and (i) $n^{-1}\sum_{i=1}^n(P_{ii}-\alpha)^2 \rightarrow 0$ or (ii) $\{(u_i,v_i)\}_{i=1,\dots,n}$ are i.i.d. normal,  as $n\rightarrow\infty$,
     \begin{equation}
         W \stackrel{d}{\rightarrow} \frac{\xi^2}{1-2\rho\frac{\xi}{\nu}+\frac{\xi^2}{\nu^2}}, \label{dis:w}
     \end{equation}
     where $\xi$ and $\nu$ are two normal random variables with means 0 and $\sqrt{\frac{1-\alpha}{2}}\frac{\mu_n^2}{\sqrt{K_n}}$, respectively, unit variances and linear correlation coefficient $\rho=\frac{\sigma_{vu}}{\sqrt{\sigma_{vv}^2\sigma_{uu}^2+\sigma_{vu}^2}}$.
\end{theorem}
Assumption (i) $n^{-1}\sum_{i=1}^n(P_{ii}-\alpha)^2 \rightarrow 0$ or (ii) i.i.d normal $\{(u_i,v_i)\}_{i=1,\dots,n}$ imposes conditions on instrument designs or errors, respectively. 
The former is known as the asymptotically balanced instruments design that is often imposed in the many IV literature, see, for example, \citep{hausman2012instrumental}, \citep{anatolyev2011specification} and \citep{wang2016bootstrap}.  We refer to \cite{anatolyev2017asymptotics} on the detailed discussions on this assumption.  Under either  Assumption (i) which implies $\omega=\alpha^2$ or Assumption (ii) which provides analytic error moments, $W$ converges in distribution to  a mixture of two normal random variables.  This result is largely different from the standard Chi-squared distribution that holds under a fixed number of instruments. It indicates that $W$ will behave close to the Chi-squared distribution only when ${\mu_n^2}/{\sqrt{K_n}}$ is unbounded. However, if ${\mu_n^2}/{\sqrt{K_n}}$ is bounded, the Chi-squared distribution will produce poor finite sample approximations and lead to size distortions. It further confirms that ${\mu_n^2}/{\sqrt{K_n}}$ is an adequate indicator for the strength of many instruments.
Therefore, based on (\ref{dis:w}), we can control the worst rejection rate
of the B2SLS-Wald test for a given tolerance level $T$:
\begin{equation*}
        \max_{\rho\in[-1,1]}\mathrm{P}\left( \frac{\xi^2}{1-2\rho\frac{\xi}{\nu}+\frac{\xi^2}{\nu^2}} \geq q_{\tau}^{\chi_1^2} \right) < T.
\end{equation*}
Using simulations,  a theoretical value of $\mu_n^2/\sqrt{K_n}$ that corresponds to the tolerance level $T$, denoted by $C$, can be determined. Consequently, the null hypothesis of many weak IVs can be formulated by $H_0:{\mu_n^2}/{\sqrt{K_n}}\leq C$, that can be tested using the $F_c$ statistic as follows:

\begin{theorem} \label{theo:fc}
     Under the assumptions of Theorem \ref{theo:wald}, as $n\rightarrow\infty$,
     \begin{equation}
        \mathrm{P}\left(F_c>\Phi^{-1}(1-\tau)\right)\rightarrow \tau. \label{Fc}
     \end{equation}
\end{theorem}

Finally, we propose the following two-step procedure to detect many weak instruments based on the first-stage $F$ statistic: 
\begin{enumerate}
    \item Obtain $C$ by controlling the worst asymptotic rejection rate of the B2SLS-Wald test at a tolerance level $T$:
    \item Use the $F_c$ test and (\ref{Fc}) to draw inference.
\end{enumerate}

One notable advantage of this two-step procedure is that the implementation of the first step is identical to that of Section 5 in MS2022, except for the different measures for the instrument strength, see discussions in Appendix. Therefore, one can directly obtain the upper bound $C$ without simulating the first-step using the relationship:
\begin{equation}
    C=\sqrt{\frac{2}{1-K_n/n}}C_0 \label{relationship}
\end{equation}
where $C_0$ is proposed to be 2.5 in MS2022. For example, if $K_n=100$ and $n=1000$ in practice, then $C=3.7$ and researcher can use the $F_c$ test to give a fast and reliable assessment of the instrument strength.

\begin{remark}
    As a byproduct of our main theorem, the proposed  $F_c$ test  is conceptually familiar and computationally simple. However, it is limited to the case of balanced instruments or normal errors and homoscedasticity. Therefore, 
    MS2022's $\widetilde{F}$ test is recommended in practice as it allows for unbalanced instruments and heteroscadasticity. Nevertherless, the $F_c$ test offers new insights for practitioners that are accustomed to reporting the $F$ statistic: it is more reliable to report the $F_c$ test statistic instead of the original $F$ test statistic when using many instruments\footnote{A minimum criterion for considering many instruments is $K_n/n \geq 0.05$, as highlighted in \cite{hansen2022econ}.}.
\end{remark}

\section{An empirical illustration: Return to education} 
\label{sec:emp}
In this section, we re-analyse the returns to education data of \cite{angrist1991does} (henceforth referred to as AK1991) using quarter of birth as an instrument for educational attainment, and construct confidence intervals for the strength of instruments. One of the specifications in the original AK1991 uses up to 180 instruments that include 30 quarter and year of birth interactions and 150 quarter and state of birth interactions. At the time of publication, the issue of weak instruments had received little attention. Later it has been widely suggested that the setup suffers from a weak instrument problem (\citealt{angrist1995split}; \citealt{bound1995problems}). MS2022 applied their proposed pre-test and  argued the instrument set is strong with the original full data. 

As the original sample size (329,509) is larger than usual for empirical research, we consider the sample size to be 0.1\% ($n=330$),  0.2\% ($n=660$), 0.5\% ($n=1650$) and  1\% 
 ($n=3300$) of the original data, more in line with the typical empirical application. We examine the specification with 180 instruments and 1530 instruments that extend the model by including the interactions among quarter and year and state of birth. We evaluate the performance of the first-stage $F$ statistic,  the $\widetilde{F}$ statistic and the $F_c$ statistic based on 1000 randomly chosen subsamples and report the results in Table  \ref{tab:emp}\footnote{A normality check with the Shapiro-Wilk test shows that the first-stage errors are plausibly normal ($p=0.68$) so that our proposed method is applicable.}. 

For the 0.1\% subsample with 180 instruments, the average $F$ statistics is 1.53, which is far below the conventional cut-off of 10. However, the average $\widetilde{F}$ statistic is 4.55, which exceeds MS2022's cutoff of 2.5. It provides an evidence that 0.1\%-scheme produces strong instruments subsamples. Our proposed $F_c$ turns out to be 2.65 ($C=5.2$ according to (\ref{relationship})), which also claims that the instrument set is strong. When the sample size increases to 660, the first-stage $F$ statistic is uninformative. While both our proposed $F_c$ test and the $\widetilde{F}$ test determine the instruments to be weak ($C=4.1$). The findings for the case of 1530 instruments are similar. In conclusion, our proposed method is informative to identify the strength of many instruments.

\begin{table}[]\renewcommand\arraystretch{0.8}
    \centering
    \begin{tabular}{cccccc}
    \hline
    \hline
           $n$ & $K_n$ & Avg.$F$ & Avg.$\Tilde{F}$   & Avg.$F_c$  \\
           \hline
         330 & 180 & 1.53 & 4.55  & 2.65   \\
         660 & 180 & 1.06 & 1.65 & 0.95   \\
        1650 & 1530 & 1.26 & 4.25 & 5.2 \\
         3300 & 1530 &1.06 & 1.35 & 0.78 \\
         \hline
         \hline
    \end{tabular}
    \caption{Empirical Results  \label{tab:emp}}
\end{table}

\section{Conclusion} \label{sec:con}
Empirical researchers often use a large number of  instruments in practice. 
In this paper, we investigate the behaviour of the first-stage $F$ statistic with many instruments. We establish that the first-stage $F$ statistic is asymptotically standard normal after appropriate normalization and recentering, which contrasts with the conventional noncentral Chi-squared distribution. We show that SY2005’s $F$ test will lead to size distortions for detecting many weak instruments, no matter which IV estimator or IV-Wald test is pretested. 

 As a byproduct of our main theory, we propose a two-step procedure for many weak instruments based on the $F$-statistic.   The proposed method is  conceptually familiar and computationally simple. This suggests that researchers can still assess the strength of many instruments relying on the $F$ statistic after proper corrections.

For future directions, it would be interesting to study the asymptotic behaviour of \cite{olea2013robust}'s effective $F$ statistic
under the many-instrument setting as it is robust to
heteroscedasticity, autocorrelation, and clustering. We conjecture that, after proper recentering and renormalizations, it would be asymptotically normal, indicating that the effective $F$ test would also have size distortions with many instruments. To establish such theoretical justifications, new tools such as the joint CLT for several sesquilinear forms under non-i.i.d. settings are needed.

\bibliographystyle{chicago}

\bibliography{Bibliography-MM-MC}

\section{Proofs of main results}\label{app}
We first prove Theorem \ref{theo:F}, then prove Corollary \ref{size_F} and Proposition  \ref{theo:seql} applying Theorem \ref{theo:F}. Finally, we prove Theorem  \ref{theo:wald}.

\subsection{Proof of Theorem \ref{theo:F}}
Denote $\boldsymbol{\Upsilon}=\Z\pi$, we first establish the following lemma:
\begin{lemma} \label{lem:F} Under Assumptions \ref{assum:1} and \ref{assum:2}, as $n\rightarrow \infty$,
$$\left(\frac{\mathbf{v}'\mathbf{v}-n\sigma_{vv}^2}{\sqrt{n}},  \frac{\mathbf{v}'\PP_Z\mathbf{v}-K_n\sigma_{vv}^2}{\sqrt{n}}\right)'\stackrel{d}{\rightarrow}N\left(\mathbf{0},(\boldsymbol{\Sigma}_{ij})_{i,j=1,2}\right),$$
with
$$\boldsymbol{\Sigma}_{11}=\mathrm{Var}(v_1^2), \; \boldsymbol{\Sigma}_{22}=\omega\mathrm{E}(v_1^4)+(2\alpha-3\omega)\sigma_{vv}^4,\; and\; \boldsymbol{\Sigma}_{12}=\alpha\mathrm{Var}(v_1^2).$$
\end{lemma}
\begin{proof}
   We apply Theorem 2 in \cite{wang2014joint} by setting $\mathbf{A}_n=\mathbf{I}_n$ and $\mathbf{B}_n=\mathbf{P}_Z$. We verify that their defined quantities $\omega_1=\theta_1=\tau_1=1$, $\theta_2=\tau_2=\omega_3=\theta_3=\tau_3=\alpha$ and $\omega_2=\omega$ exists since the leverage value $P_{ii}$ ranges in $[0,1]$. One can verify $A_1=\mathrm{Var}(v_1^2)$,  $A_2=A_3=\sigma_{vv}^4$ in Theorem 2 of \cite{wang2014joint}, which completes the proof.
\end{proof}

To prove Theorem \ref{theo:F}, firstly note that
\begin{equation*}
    F=\frac{n-K_n}{K_n} \frac{\frac{\boldsymbol{\Upsilon}'\boldsymbol{\Upsilon}}{n}+2\frac{\mathbf{v}'\boldsymbol{\Upsilon}}{n}+\frac{\mathbf{v}'\PP_Z\mathbf{v}}{n}}{\frac{\mathbf{v}'\mathbf{v}}{n}-\frac{\mathbf{v}'\PP_Z\mathbf{v}}{n}}.
\end{equation*}
Besides,
\begin{equation*}
    \frac{\mathbf{v}'\boldsymbol{\Upsilon}}{n}=o_p(\frac{1}{\sqrt{n}})
\end{equation*}
as
\begin{equation*}
    \mathrm{E}\left(\frac{\mathbf{v}'\boldsymbol{\Upsilon}}{n}\right)^2=\mathrm{E} \left(\frac{\boldsymbol{\Upsilon}'\mathbf{v}\mathbf{v}'\boldsymbol{\Upsilon}}{n^2}\right)= \sigma_{vv}^2\frac{\boldsymbol{\Upsilon}'\boldsymbol{\Upsilon}}{n^2}=o(\frac{1}{n}).
\end{equation*}
We then can apply Delta method with Lemma \ref{lem:F} and the function $f:\mathbb{R}^2\rightarrow\mathbb{R}$ satisfying that $$f\left(\frac{\mathbf{v}'\mathbf{v}}{n},  \frac{\mathbf{v}'\PP_Z\mathbf{v}}{n}\right)=F.$$
We have $$\nabla f\left(\sigma_{vv}^2,\frac{K_n}{n}\sigma_{vv}^2\right)=\left(-\frac{n\sum_{i=1}\Upsilon_i^2+nK_n\sigma_{vv}^2}{K_n(n-K_n)\sigma_{vv}^4},\frac{n\sum_{i=1}\Upsilon_i^2+n^2\sigma_{vv}^2}{K_n(n-K_n)\sigma_{vv}^4}\right),$$
which yields with Lemma \ref{lem:F} that
$$\sqrt{n}(F-1-\frac{\mu_n^2}{K_n })\stackrel{d}{\rightarrow}N\left(0,\frac{(\omega-\alpha^2)\mathrm{E}(v_1^4)+(2\alpha-3\omega+\alpha^2)\sigma_{vv}^4}{\alpha^2(1-\alpha)^2\sigma_{vv}^4}\right).$$

\subsection{Proof of Corollary \ref{size_F}}
By the noncentral Chi-squared distribution's normal approximation for large $K_n$, we verify that
\begin{equation*}
    q_{\tau}^{\chi^2_{K_n}(\mu_0^2)}=\frac{K_n+\mu_0^2}{2K_n+4\mu_0^2}\left(\Phi^{-1}(1-\tau)+\sqrt{2K_n+4\mu_0^2}\right)^2+O(1),
\end{equation*}
 so 
\begin{equation*}
    \frac{q_{\tau}^{\chi^2_{K_n}(\mu_0^2)}}{K_n}-1=\frac{1+\mu_0^2/{K_n}}{2K_n+4\mu_0^2}(\Phi^{-1}(1-\tau))^2+\frac{2+2\mu_0^2/K_n}{\sqrt{2K_n+4\mu_0^2}}\Phi^{-1}(1-\tau)+\frac{\mu_0^2}{K_n}+O(\frac{1}{K_n}).                       
\end{equation*}        
Denote $$\sigma_F^2=\frac{(\omega-\alpha^2)\mathrm{E}(v_1^4)+(2\alpha-3\omega+\alpha^2)\sigma_{vv}^4}{\alpha^2(1-\alpha)^2\sigma_{vv}^4}.$$Then, using Theorem \ref{theo:F}, 
\begin{align*}
    \mathbb{P}\left( K_nF > q_{\tau}^{\chi^2_{K_n}(\mu_0^2)}\right)&=\mathbb{P}\left( {\sigma}^{-1}_F\sqrt{n}(F-1-\frac{\mu_0^2}{K_n}) >{\sigma}^{-1}_F\sqrt{n}(\frac{q_{\tau}^{\chi^2_{K_n}(\mu_0^2)}} {K_n}-1-\frac{\mu_0^2}{K_n})\right)\nonumber\\&=\mathbb{P}\left({\sigma}^{-1}_F\sqrt{n}(F-1-\frac{\mu_0^2}{K_n})> {\sigma}^{-1}_F\sqrt{\frac{2}{\alpha}}\Phi^{-1}(1-\tau)+o(1)\right)\\&\rightarrow1-\Phi\left(\sqrt{\frac{2}{\alpha{\sigma}^{2}_F}}\Phi^{-1}(1-\tau)\right). \label{eq:si}
\end{align*}
Especially, when $\omega=\alpha^2$ or $v_i$ has zero excess kurtosis, then $\sigma_F^2=2/(\alpha(1-\alpha))$, so that 
\begin{align*}
    \mathbb{P}\left( K_nF > q_{\tau}^{\chi^2_{K_n}(\mu_0^2)}\right)\rightarrow1-\Phi\left(\sqrt{1-\alpha}\Phi^{-1}(1-\tau)\right)>\tau.
\end{align*}

\subsection{Proof of Proposition \ref{theo:seql}}
Similar to the proofs of Theorem \ref{size_F}, we have
\begin{align*}
    \mathbb{P}\left( K_nF > q_{\tau}^{\chi^2_{K_n}(\mu_0^2)}\right)&=\mathbb{P}\left( {{2}}^{-1/2}\sqrt{K_n}(F-1-\frac{\mu_0^2}{K_n}) >{{2}}^{-1/2}\sqrt{K_n}(\frac{q_{\tau}^{\chi^2_{K_n}(\mu_0^2)}} {K_n}-1-\frac{\mu_0^2}{K_n})\right)\nonumber\\&=\mathbb{P}\left({{2}}^{-1/2}\sqrt{K_n}(F-1-\frac{\mu_0^2}{K_n})> \Phi^{-1}(1-\tau)+o(1)\right)\\&\rightarrow1-\Phi\left(\Phi^{-1}(1-\tau)\right)=\tau, 
\end{align*}
where the last convergence holds by the CLT under the SEQ-L: $\sqrt{K_n}(F-1-\mu_0^2/K_n)\stackrel{d}{\rightarrow}N(0,2)$.

\subsection{Proof of Theorem \ref{theo:wald}}
Define
\begin{equation*}
    Q_{Yu}=\frac{\Y'\PP_b\uu}{\sqrt{K_n}}, \; Q_{YY}=\frac{\Y'\PP_b\Y}{\sqrt{K_n}}.
\end{equation*}
We first introduce two lemmas: 
\begin{lemma} \label{lem1:wald}Suppose that $\mu_n^2=o(n)$, under Assumption of the  Theorem \ref{theo:wald}, as $n \rightarrow \infty$,
  $$\left( Q_{Yu}, Q_{YY} - \frac{(1-\frac{K_n}{n})\pi'\Z'\Z\pi}{\sqrt{K_n}}\right)\stackrel{d}{\rightarrow} N(0,\Sigma),$$ where $\Sigma$ is the asymptotic covariance matrix, with  elements:
  \begin{equation*}\label{rho}
       \sigma_1^2=(1-\alpha)(\sigma_{vv}^2\sigma_{uu}^2+\sigma_{vu}^2), \; \sigma_2^2=2(1-\alpha)\sigma_{vv}^4, \; \sigma_{12}=2(1-\alpha)\sigma_{vv}^2\sigma_{vu}, \;  \rho=\frac{\sigma_{12}}{\sigma_{1}\sigma_{2}}.
  \end{equation*}
\end{lemma}
\begin{proof}
    From Theorem 2 in \cite{wang2014joint}, we obtain that
     $$\left( \frac{\vv'\PP_b\uu}{\sqrt{K_n}}, \frac{\vv'\PP_b\vv}{\sqrt{K_n}}\right)\stackrel{d}{\rightarrow} N(0,\Sigma).$$
     The proof is then completed by noticing that $\mu_n^2=o(n)$.
\end{proof}

\begin{lemma}\label{lem2:wald}Suppose that $\mu_n^2=o(n)$, under the assumptions of  Theorem \ref{theo:wald},
\begin{itemize}
    \item[(i).] $\frac{\Y'\mathbf{M}_Z\Y}{n}=(1-\alpha)\sigma_{vv}^2+o_p(1)$,
     \item[(ii).] $\frac{\hat{\uu}'\mathbf{M}_Z\hat{\uu}}{n}=(1-\alpha)\sigma_{uu}^2-2(1-\alpha)\sigma_{vu}\frac{Q_{Yu}}{Q_{YY}}+(1-\alpha)\sigma_{vv}^2\frac{Q_{Yu}^2}{Q_{YY}^2}+o_p(1)$,
     \item[(iii).] $\frac{\Y'\mathbf{M}_Z\hat{\uu}}{n}=(1-\alpha)\sigma_{vu}-(1-\alpha)\sigma_{vv}^2\frac{Q_{Yu}}{Q_{YY}}+o_p(1)$.
\end{itemize}
\end{lemma}
\begin{proof}
    Proof of (i) can be proved using standard arguments for quadratic forms by noting that $\Y'\mathbf{M}_Z\Y=\vv'\mathbf{M}_Z\vv$. 

    To prove (ii), note that
    \begin{align*}
        \frac{\hat{\uu}'\mathbf{M}_Z\hat{\uu}}{n}=\frac{\uu'\mathbf{M}\uu}{n}-2\frac{\vv'\mathbf{M}\uu}{n}\frac{Q_{Yu}}{Q_{YY}}+\frac{\Y'\mathbf{M}_Z\Y}{n}\frac{Q_{Yu}^2}{Q_{YY}^2}.
    \end{align*}
    the proof is completed noticing that ${\uu}'\mathbf{M}_Z{\uu}/n\stackrel{p}{\rightarrow}(1-\alpha)\sigma_{uu}^2$ and $\vv\mathbf{M}_Z\uu/n\stackrel{p}{\rightarrow}(1-\alpha)\sigma_{vu}$. The proof of (iii) is similar so we omit it here.
\end{proof}
The B2SLS-Wald test statistic can be rewritten to:
\begin{align*}
    W&=\frac{Q_{Yu}^2}{Q_{YY}^2n^{-1}\hat{V}}
\end{align*}
where the denominator expands to
\begin{align*}
    \frac{n-K_n}{K_n}\frac{\Y'\PP_b\Y}{K_n}\hat{\sigma}_{uu}^2+\frac{n}{n-K_n}\left(\frac{\hat{\uu}'\mathbf{M}_Z\hat{\uu}}{n} \frac{\Y'\mathbf{M}_Z\Y}{n}+\left(  \frac{\Y'\mathbf{M}_Z\hat{\uu}}{n}\right)^2\right).
\end{align*}
From Lemma \ref{lem2:wald}, this expansion further converges to $\sigma_1^2-2\sigma_{12}{Q_{Yu}}/{Q_{YY}}+\sigma_2^2Q_{Yu}^2/Q_{YY}^2$, so that
\begin{align*}
    W=\frac{Q_{Yu}^2}{\sigma_1^2-2\sigma_{12}\frac{Q_{Yu}}{Q_{YY}}+\sigma_2^2\frac{Q_{Yu}^2}{Q_{YY}^2}}(1+o_p(1))=\frac{\xi^2}{1-2\rho\frac{\xi}{\nu}+\frac{\xi^2}{\nu^2}}(1+o_p(1)).
\end{align*}

\appendix
\newpage
\begin{center}
    \Large{\textbf{Supplementary Materials \protect \\ for}}  \\ \large{\textbf{``The First-stage F Test with Many Weak Instruments''}} 
\end{center}

\numberwithin{equation}{section}

\section{Concentration parameter, $F_c$ and $\widetilde{F}$}
MS2022 defined the jackknifed concentration parameter as follows:
\begin{equation*} 
    \Tilde{\mu}_n^2=\Gamma^{-1/2}\pi'\Z'(\PP_Z-\mathbf{D}_Z)\Z\pi, 
\end{equation*}
where $\Gamma$ is an unknown variance term and proposed the $\widetilde{F}$ test statistic to measure it. Specifically, MS2022 derived that  the asymptotic distribution of the JIVE-Wald test statistic is 
\begin{equation} \label{wald:jive}
    W_{JIVE} \stackrel{d}{\rightarrow} \frac{\xi_0^2}{1-2\rho_0\frac{\xi_0}{\nu_0}+\frac{\xi_0^2}{\nu_0^2}},
\end{equation}
where $\xi_0$ and $\nu_0$ are two normal random variables with means 0 and $\frac{\Tilde{\mu}_n^2}{\sqrt{K_n}}$, unit variances and a correlation coefficient $\rho_0$. They proposed a two-step procedure as following:
\begin{itemize}
    \item[1.] Obtain $C_0$ by controlling the worst asymptotic rejection rate of the JIVE-Wald test at a tolerance level $T$:
    \begin{equation*}
        \max_{\rho_0\in[-1,1]}\mathrm{P}\left( \frac{\xi^2}{1-2\rho_0\frac{\xi_0}{\nu_0}+\frac{\xi_0^2}{\nu_0^2}} \geq q_{\tau}^{\chi_1^2} \right) < T.
    \end{equation*}
    \item[2.] Use the proposed $\widetilde{F}$ statistic, which asymptotically follows $N({\Tilde{\mu}_n^2}/{\sqrt{K_n}},1)$, to make inference.
\end{itemize}

As the maximum rejection rate to the Wald test occurs at $\rho_0=1$, it follows that the behaviour of B2SLS-Wald and JIVE-Wald share the same pattern apart from the mean of the second normal random variable, i.e., the measure of the instrument strength. We observe that when controlling the worst rejection rate at a given level, the obtained $C_0$ is equivalent to the desired $\sqrt{\frac{1-\alpha}{2}}\frac{\mu_n^2}{\sqrt{K_n}}$. Consequently, our proposed first-step can be implemented without simulation studies based on the relationship between $C$ and $C_0$:
$C=\sqrt{\frac{2}{1-K_n/n}}C_0$.

\section{Collapse of SY2005's first step} 
Stock-Yogo's procedure relies on the established relationship between the relative bias and the concentration parameter. In this section, we show that this relationship  no longer holds when the number of instruments is large. Let us consider the following linear IV model with multiple endogenous variables:
\begin{equation} \label{model:p1}
    \underset{ n\times 1}{\y}=\underset{ n\times p}{\Y}\;\underset{ p\times 1}{\bbeta}+\underset{ n\times 1}{\uu},
\end{equation}
\begin{equation} \label{model:p2}
    \underset{ n\times p}{\Y}=\underset{ n\times K_n}{\Z}\underset{ K_n\times p}\PPi+\underset{ n\times p}{\V},
\end{equation}
where we still use $\Y$ for simplicity to denote the observations on the endogenous variables, and the $K_n \times p$ matrix $\PPi$ is the first-stage coefficients. We assume the $(u_i,\V_i)'$ are i.i.d. multivariate distributed with zero mean and $\mathrm{cov}(\V_i)=\boldsymbol{\Sigma}_{VV}$, $\mathrm{var}(u_i)=\sigma_u^2$ and $\mathrm{cov}(u_i,\V_i)=\boldsymbol{\Sigma}_{uV}$. 
As a measure of the relative bias of 2SLS, SY2005 proposed the ratio
\begin{equation}
    B_n^2=\frac{\mathrm{E}(\hat{\beta}^{2SLS}-\beta)'{\Sigma}_{YY}\mathrm{E}(\hat{\beta}^{2SLS}-\beta)}{\mathrm{E}(\hat{\beta}^{OLS}-\beta)'{\Sigma}_{YY}\mathrm{E}(\hat{\beta}^{OLS}-\beta)},
\end{equation}
where $\Sigma_{YY}= {\mathrm{plim}_{n\rightarrow \infty}}(\Y'\Y/n)$. They showed that 
\begin{equation} \label{SY_B}
    B_n\rightarrow B_{SY}= \left|\mathrm{E}\left[ \frac{(\lambda+\xi_s)'\xi_s}{(\lambda+\xi_s)'(\lambda+\xi_s)} \right] \right|,
\end{equation}
where $\lambda'\lambda=\underset{ n}{\lim}\mu_n^2$ and $\xi_s\sim N_{K}(\mathbf{0},\mathbf{I}_{K})$. It follows that the concentration parameter can fully characterize the asymptotic relative bias.\footnote{\cite{skeels2018stock} showed that $B_{SY}$ is a strictly decreasing continuous function of $\mu_0^2$.}  The instruments are deemed to be weak if  the  asymptotic relative bias is larger than a predetermined tolerance level.

To study the behavior of $B_n^2$ with many instruments, we impose the following assumption on  the 2SLS and OLS estimators.
\begin{assumption} \label{assum:ui}
    Both $\hat{\bbeta}^{2SLS}$ and $\hat{\bbeta}^{OLS}$ are uniformly integrable.
\end{assumption}
This assumption is of a high level, designed to ensure that the convergence in probability of $\hat{\bbeta}^{2SLS}-\hat{\bbeta}^{OLS}$ is indicative of its convergence in mean. A study conducted by \cite{afendras2016uniform} proved that under regularity conditions on $\Y$, $\hat{\bbeta}^{OLS}$ is uniformly integrable. We posit that $\hat{\bbeta}^{2SLS}$ is also uniformly integrable, given regularity conditions on $\Z$, and defer the proof to future research.
The behavior of the 2SLS relative bias is characterized in the following theorem.
\begin{theorem} \label{theo:bias} 
Under Assumptions \ref{assum:1}, \ref{assum:2} and \ref{assum:ui}, suppose that $\PPi'\Z'\Z\PPi/n\rightarrow \boldsymbol{\Theta} $ almost surely, as $n\rightarrow\infty$, we have almost surely,
\begin{equation} \label{eq:b}
     B_n^2 \rightarrow B^2=\frac{\alpha^2\boldsymbol{\Sigma}_{uV} \boldsymbol{\Theta}_1^{-1} \boldsymbol{\Theta}_2 \boldsymbol{\Theta}_1^{-1}\boldsymbol{\Sigma}_{Vu}}{\boldsymbol{\Sigma}_{uV} \boldsymbol{\Theta}_2^{-1} \boldsymbol{\Sigma}_{Vu}},
\end{equation}
 where $\boldsymbol{\Sigma}_{uV} =\boldsymbol{\Sigma}_{Vu}'$,  $\boldsymbol{\Theta}_1=\boldsymbol{\Theta}+\alpha\boldsymbol{\Sigma}_{VV}$ and $\boldsymbol{\Theta}_2=\boldsymbol{\Theta}+\boldsymbol{\Sigma}_{VV}$. In particular, if $\boldsymbol{\Theta}=0$, $ B^2=1\label{relative:null}$, otherwise, $0<B^2<1$.
\end{theorem}
Theorem \ref{theo:bias} shows that the relative bias will converge to a positive constant when the number of instrument is large. The limit is no more than one and it equals one only when the concentration parameter has the order smaller than $n$. Figure \ref{fig:bias}(a) shows the asymptotic relative bias $B_{SY}$ in (\ref{SY_B}) as a function of $\mu_0^2$ when $p=1$ and $K=5$. When the number of instruments is large, Figure \ref{fig:bias}(b) presents the asymptotic relative bias in (\ref{eq:b}) as a function of $\mu_0^2$, where the x-axis plots the order of $\mu_0^2$.  The relationship between $B_{SY}$ and $\mu_0^2$ given in (\ref{SY_B}) for a small number of instruments no longer holds for the case of many instruments.  Furthermore, as long as $\PPi'\Z'\Z\PPi=o(n)$, the relative bias will converge to one. This finding indicates that testing for weak instruments based on the relative bias becomes conceptually unimplementable with a large number of instruments as 2SLS will always have the same level of bias as OLS does. 

\begin{figure}[h!]
\centering
\subfloat[$B_{SY}$ as a function of $\mu^2$ given in (\ref{SY_B})]{\includegraphics[width=8.1cm]{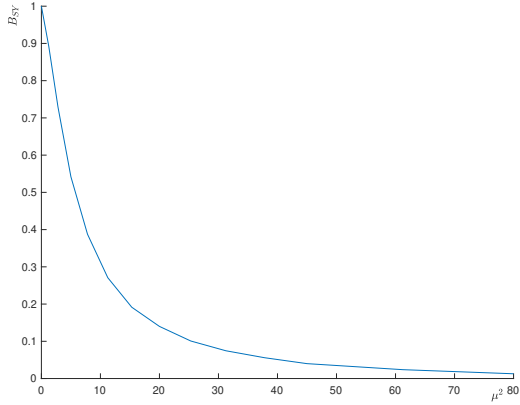}}\hfil
\subfloat[$B^2$ as a function of $\mu^2$ given in (\ref{eq:b})]{\includegraphics[width=8.1cm]{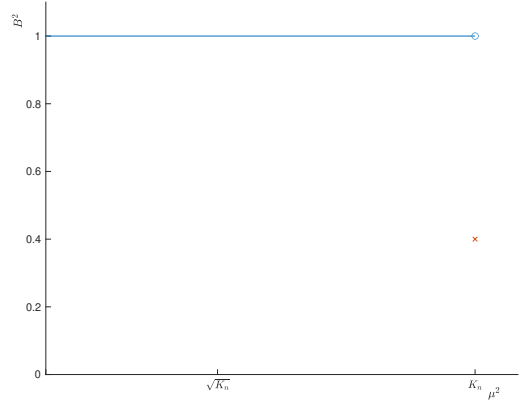}}\hfil
\caption{The plotted relative bias as a function of the concentration parameter for the case of small number of instruments ($K_n=5$) and large number of instruments ($K_n=500$), respectively. The sample size is $n=1000$.} \label{fig:bias}
\end{figure}

\section{Multiple endogenous variables} \label{sec:cd}
In this section, we first define the weak instruments in the context of multiple endogenous variables. Next, we derive the limiting distribution of the trace of the Cragg-Donald statistic in the linear IV model with multiple endogenous variables, as defined in Equation (\ref{model:p1}) and (\ref{model:p2}), under the many instruments setup. Building on this result, we propose a explicit procedure for assessing the strength of many instruments when there are two endogenous variables.

MS2022 defines weak instruments with a single endogenous variable by the fact that no consistent estimators and tests exist when $\mu_n^2/\sqrt{K_n}$ is bounded. We adopt the same idea and characterize the strength of instruments with multiple endogenous variables as follows:
\begin{equation*}
    \theta_n^2={\mathrm{tr}(\boldsymbol{\Sigma}_{VV}^{-1/2}\PPi'\Z'\Z\PPi\boldsymbol{\Sigma}_{VV}^{-1/2})},
\end{equation*}
which reduces to $\mu_n^2$ in the case of $p=1$.  We define weak instruments in the context of multiple endogenous variables when $\theta_n^2/\sqrt{K_n}$ is bounded on the grounds that consistent estimation is not achievable in this scenario \citep{chao2005consistent}.

We now consider testing for many weak instruments. For the case of a fixed number of instruments, Stock-Yogo proposed to use the minimum eigenvalue of the following Cragg-Donald statistic, which is the matrix analog of the first-stage $F$-statistic: 
\begin{equation} \label{cd}
    \mathbf{CD}=\frac{n-K_n}{K_n}(\Y'\mathbf{M}_{Z}\Y)^{-1}\Y'\PP_{Z}\Y .
\end{equation}
The obtained test is documented to be conservative by applying a Chi-squared bound to the noncentral Wishart distribution. Stock-Yogo also found that the behavior of the test procedure depends on all  eigenvalues of the Cragg-Donald statistic through the relative bias when the number of instruments increases. That is, when $p\geq 1$, the strength of instruments is contained in all eigenvalues of the concentration matrix $\boldsymbol{\Sigma}_{VV}^{-1/2}\PPi'\Z'\Z\PPi\boldsymbol{\Sigma}_{VV}^{-1/2}$.
This finding motivates us to investigate the asymptotic behavior of the trace of (\ref{cd}) in case of many instruments.

\begin{theorem} \label{theo:cd}  When $\theta_n^2/K_n\rightarrow0$, under Assumptions \ref{assum:1} and \ref{assum:2}, as $n\rightarrow\infty$, we have
\begin{equation*}
    \sqrt{n}\left(\mathrm{tr}(\mathbf{CD})-p-\frac{\theta_n^2}{K_n}\right) \stackrel{d}{\rightarrow} N(0,\sigma_0^2),
\end{equation*}
where $\sigma_0^2=(1-\alpha)^{-2}\mathrm{vec}'(\boldsymbol{\Sigma}_{VV}^{-1})\left(\alpha^{-2}\boldsymbol{\Sigma}_4-\boldsymbol{\Sigma}_3\right)\mathrm{vec}(\boldsymbol{\Sigma}_{VV}^{-1})$ with matrices $\boldsymbol{\Sigma}_1$ and $\boldsymbol{\Sigma}_2$ defined in (\ref{eq:sigma11}) and (\ref{eq:sigma22}) in Appendix, respectively.
\end{theorem}

Theorem \ref{theo:cd} establishes the asymptotic normality of the trace of the Cragg-Donald statistic with many instruments. This result holds for a general number of endogenous variables and the errors can be nonnormal. Unfortunately, the limiting variance $\sigma_0^2$ has a complex form. It can be nonaccessible from data, see Remark \ref{rem:1}. This problem vanishes for the case of two endogenous variables and we have the following result.
\begin{corollary} \label{p2} For $p=2$, when $\theta_n^2/K_n\rightarrow0$, under Assumptions \ref{assum:1}, \ref{assum:2} and \ref{Fassum:3}, as $n\rightarrow\infty$,
\begin{equation*}
   \sqrt{n}\left(\mathrm{tr}(\mathbf{CD})-2-\frac{\theta_n^2}{K_n}\right) \stackrel{d}{\rightarrow} N\left(0,\frac{4}{\alpha(1-\alpha)}\right).
\end{equation*}
\end{corollary}
Hence, when there are two endogenous variables, Corollary \ref{p2} allows us to construct confidence intervals for the strength of instruments and propose an additional test statistic for 
\begin{equation} \label{null1}
    H_0^{(2)}:\frac{\theta_n^2}{\sqrt{K_n}}\leq C 
\end{equation} using the fact that
\begin{equation}
    CD_{c}=\sqrt{\frac{K_n(n-K_n)}{4n}}\left[\mathrm{tr}(\mathbf{CD})-2-\frac{C}{\sqrt{K_n}}\right],
\end{equation}
 is asymptotically standard normal under the equality in $H_0^{(2)}$.

\begin{remark}\label{rem:1}
Obtaining the explicit form of $\sigma_0^2$ is tedious. For example, when $p=3$, if $\mathbf{v}_i\stackrel{i.i.d}{\sim}N_3\left(\mathbf{0},\mathrm{diag}(  \sigma_{1}^2,\sigma_{2}^2,\sigma_{3}^2)\right)$, then the limiting variance is ${6}/\left({\alpha(1-\alpha)}\right)$ under the same assumptions. Generally, the limiting variance will depend on $\mathrm{E}(v_{i1}^{m_1}v_{i2}^{m_2}v_{i3}^{m_3})$ with non-negative integers $m_1$, $m_2$ and $m_3$ satisfying $m_1+m_2+m_3=4$ and the covariances between $v_{i1}$, $v_{i2}$ and $v_{i3}$.
\end{remark}

\begin{remark}
    In practical applications, we recommend the  range of $[1,3]$ as a conservative interval indicative of many weak instruments when there are multiple endogenous variables.  Most of the simulation settings for weak instruments in literature corresponds to values of $C$ that are also consistent with the proposed range [1,3] (e.g. \citealt{wang1998inference}, \citealt{hansen2008estimation}, \citealt{hausman2012instrumental} and \citealt{wang2016bootstrap}). However, we acknowledge that determining the upper bound $C$ in equation (\ref{null1}) should be done separately. Developing a two-step procedure within the framework of many instruments and multiple endogenous variables still remains an open research question in the existing literature. Further investigation and advancements in this area are necessary to address this challenge effectively.
\end{remark}

\section{Proofs}\label{s: app}

\subsection{Proof of Theorem \ref{theo:bias}}
We first introduce the following lemma which establishes the probabilistic limit of 2SLS and OLS within the many instruments setup.
\begin{lemma}(Theorem 1, \citealt{huang2022specification})
    Under Assumptions \ref{assum:1} and \ref{assum:2}, as $n\rightarrow\infty$,
(a) when $\PPi'\Z'\Z\PPi/n \rightarrow \boldsymbol{\Theta}$  almost surely, $\hat{\bbeta}^{2SLS}  \stackrel{p}{\rightarrow} \bbeta+ \alpha\boldsymbol{\Theta}_1^{-1}\boldsymbol{\Sigma}_{Vu}$, $\hat{\bbeta}^{OLS} \stackrel{p}{\rightarrow} \bbeta+ \boldsymbol{\Theta}_2^{-1}\boldsymbol{\Sigma}_{Vu}$; (b) when $\PPi'\Z'\Z\PPi/n \rightarrow \boldsymbol{0}$  almost surely, both $\hat{\bbeta}^{2SLS}$ and $\hat{\bbeta}^{OLS}$ converge in probability to $ \bbeta+ \boldsymbol{\Sigma}_{VV}^{-1}\boldsymbol{\Sigma}_{Vu}$.\label{theo:limit} 
\end{lemma}

To prove Theorem \ref{theo:bias}, note that $\boldsymbol{\Sigma}_{YY}=\boldsymbol{\Theta}_2$, we observe that results of convergence of $B_n^2$  holds naturally by applying Theorem \ref{theo:limit} directly. Now, we show that $B^2<1$ when $\PPi'\Z'\Z\PPi/n \rightarrow \boldsymbol{0}$. Define the degrees of simultaneity $\boldsymbol{\rho}=\boldsymbol{\Theta}_2^{-1/2}\boldsymbol{\Sigma}_{Vu}$, then 
\begin{equation*}
B^2=\frac{\alpha^2\boldsymbol{\rho}'\boldsymbol{\Theta}_2^{1/2}\boldsymbol{\Theta}_1^{-1}\boldsymbol{\Theta}_2\boldsymbol{\Theta}_1^{-1}\boldsymbol{\Theta}_2^{1/2}\boldsymbol{\rho}}{\boldsymbol{\rho}'\boldsymbol{\rho}}.
\end{equation*}
Similar to \cite{stock2002testing}, we consider the worst-case asymptotic relative bias
\begin{equation*}
    B_{\max }^2=\max _{\boldsymbol{\rho}}|B|^2.
\end{equation*}
It can be shown that $B_{\max }^2=\lambda_1(\alpha^2\boldsymbol{\Theta}_2^{1/2}\boldsymbol{\Theta}_1^{-1}\boldsymbol{\Theta}_2\boldsymbol{\Theta}_1^{-1}\boldsymbol{\Theta}_2^{1/2})$. Note that $\boldsymbol{\Theta}_1/\alpha-\boldsymbol{\Theta}_2= (1-\alpha)/\alpha \boldsymbol{\Theta}$ is positive definite, we have
\begin{align}
     \alpha\boldsymbol{\Theta}_2^{1/2}\boldsymbol{\Theta}_1^{-1}\boldsymbol{\Theta}_2^{1/2}&=\boldsymbol{\Theta}_2^{1/2}(\frac{1-\alpha}{\alpha }\boldsymbol{\Theta}+\boldsymbol{\Theta}_2)^{-1}\boldsymbol{\Theta}_2^{1/2} \nonumber
     \\&=\boldsymbol{\Theta}_2^{1/2}\left(\boldsymbol{\Theta}_2^{-1}-\boldsymbol{\Theta}_2^{-1}(\boldsymbol{\Theta}_2^{-1}+ \frac{\alpha}{1-\alpha }\boldsymbol{\Theta}^{-1})^{-1}\boldsymbol{\Theta}_2^{-1} \right)\boldsymbol{\Theta}_2^{1/2}\nonumber \\&=
\mathrm{I}_p-\boldsymbol{\Theta}_2^{-1/2}(\boldsymbol{\Theta}_2^{-1}+ \frac{\alpha}{1-\alpha }\boldsymbol{\Theta}^{-1})^{-1}\boldsymbol{\Theta}_2^{-1/2},
\end{align}
where the second equality follows from the Woodbury matrix identity. Next, by Weyl's inequality, it follows that
\begin{equation*}
    \lambda_1(\alpha\boldsymbol{\Theta}_2^{1/2}\boldsymbol{\Theta}_1^{-1}\boldsymbol{\Theta}_2^{1/2})+\lambda_p\left(\boldsymbol{\Theta}_2^{-1/2}(\boldsymbol{\Theta}_2^{-1}+\frac{\alpha}{1-\alpha }\boldsymbol{\Theta}^{-1})^{-1}\boldsymbol{\Theta}_2^{-1/2}\right) \leq 1.
\end{equation*}
Note that $\boldsymbol{\Theta}_2^{-\frac{1}{2}}(\boldsymbol{\Theta}_2^{-1}+ \frac{\alpha}{1-\alpha }\boldsymbol{\Theta}^{-1})^{-1}\boldsymbol{\Theta}_2^{-\frac{1}{2}}$ is positive definite, then $\lambda_1(\alpha^2\boldsymbol{\Theta}_2^{1/2}\boldsymbol{\Theta}_1^{-1}\boldsymbol{\Theta}_2\boldsymbol{\Theta}_1^{-1}\boldsymbol{\Theta}_2^{1/2})\leq \lambda_1^2(\alpha\boldsymbol{\Theta}_2^{1/2}\boldsymbol{\Theta}_1^{-1}\boldsymbol{\Theta}_2^{1/2})<1$.

\subsection{Proof of Theorem \ref{theo:cd}}
To prove Theorem \ref{theo:cd}, we introduce Lemma \ref{lem:4comp}, which establishes the joint CLT of four
sesquilinear forms that make up the Cragg-Donald statistic, and then apply Delta method to it.

We firstly define the following variables:
\begin{equation} \label{v1}
\Dot{\V}:=\left[\underbrace{\overbrace{\V(1),\cdots,\V(1)}^{p},\cdots,\overbrace{\V(p),\cdots,\V(p)}^p}_{p^2}\right] \in \mathbb{R}^{n \times p^2},
\end{equation}
and
\begin{equation} \label{v2}
    \Ddot{\V}:=\left[\underbrace{\overbrace{\V(1),\cdots,\V(p)}^{p},\cdots,\overbrace{\V(1),\cdots,\V(p)}^p}_{p^2}\right] \in \mathbb{R}^{n \times p^2}.
\end{equation} 
We then establish the joint distribution of two key components which make up the Cragg-Donald statistic in the following lemma:
\begin{lemma}\label{lem:4comp}
Under Assumption \ref{assum:1} and \ref{assum:2}, as $n\rightarrow \infty$, we have
$$
\sqrt{n}\cdot\left[\left(\begin{array}{l}
n^{-1} \mathrm{vec}(\V'\V) \\
n^{-1} \mathrm{vec}(\V'\PP_Z\V)
\end{array}\right)-\left(\begin{array}{c}
\mathrm{vec}(\boldsymbol{\Sigma}_{VV})\\
\frac{K_n}{n} \mathrm{vec}(\boldsymbol{\Sigma}_{VV})
\end{array}\right)\right] \stackrel{d}{\rightarrow} N\left(\mathbf{0}, \boldsymbol{\Sigma}_0\right),
$$
where $$\boldsymbol{\Sigma}_0=\left[\begin{array}{cc}
\boldsymbol{\Sigma}_1 & \alpha \boldsymbol{\Sigma}_1 \\
\alpha \boldsymbol{\Sigma}_1  & \boldsymbol{\Sigma}_2
\end{array}\right]$$ with
\begin{equation} \label{eq:sigma11}
    \boldsymbol{\Sigma}_{1,ij}=\mathrm{Cov}(\Dot{V}_{1i}\Ddot{V}_{1i},\Dot{V}_{1j}\Ddot{V}_{1j}), 
\end{equation}
\begin{equation} \label{eq:sigma22}
    \boldsymbol{\Sigma}_{2,ij}=\omega \boldsymbol{\Sigma}_{11,ij}+(\alpha-\omega)\left(\mathrm{E}(\Dot{V}_{1i}\Dot{V}_{1j})\mathrm{E}(\Ddot{V}_{1i}\Ddot{V}_{1j})+\mathrm{E}(\Dot{V}_{1i}\Ddot{V}_{1j})\mathrm{E}(\Ddot{V}_{1j}\Ddot{V}_{1i})\right). 
\end{equation}
\end{lemma}
\begin{proof}
    We apply Theorem 2 in \cite{wang2014joint} to $\Dot{\V}$ and $\Ddot{\V}$ defined in (\ref{v1}) and (\ref{v2}), respectively, by setting $\mathbf{A}_n=\mathbf{I}_n$ and $\mathbf{B}_n=\mathbf{P}_Z$. Again we have the defined quantities $\omega_1=\theta_1=\tau_1=1$, $\theta_2=\tau_2=\omega_3=\theta_3=\tau_3=\alpha$ and $\omega_2=\omega$. The limiting covariance matrix then turned out to be $\boldsymbol{\Sigma}_{1}$ and $\boldsymbol{\Sigma}_{2}$.
\end{proof}

\noindent\textit{Proof of Theorem \ref{theo:cd}}: 
Note that
\begin{equation} \label{eq:cd}
    \mathbf{CD}= \frac{n-K_n}{K_n}\left(\frac{\PPi'\Z'\Z\PPi}{n}+\frac{\PPi'\Z'\V}{n}+\frac{\V'\Z\PPi}{n}+\frac{\V'\PP_{Z}\V}{n}\right)\left(\frac{\V'\V}{n}-\frac{\V'\PP_Z\V}{n} \right)^{-1}.
\end{equation}
Besides,
\begin{equation*}
    \frac{\V'\Z\PPi}{n}=o_p\left(\frac{1}{\sqrt{n}}\right), 
\end{equation*}
as
\begin{equation*}
    \mathrm{E}\left\| \frac{\V'\Z\PPi}{n}\right\|^2=\mathrm{E} \left[\mathrm{tr}\left(\frac{\PPi'\Z'\V\V'\Z\PPi}{n^2}\right)\right]= \mathrm{tr}(\boldsymbol{\Sigma}_{VV})\mathrm{E} \left[\mathrm{tr}\left(\frac{\PPi'\Z'\Z\PPi}{n^2}\right)\right]=o\left(\frac{1}{n}\right).
\end{equation*}
We then apply Delta method with $f:\mathbb{R}^{2p^2}\rightarrow\mathbb{R}$ satisfying that 
\begin{equation} \label{cd:delta}
    f\left(\begin{array}{l}
n^{-1} \mathrm{vec}(\V'\V) \\
n^{-1} \mathrm{vec}(\V'\PP_Z\V)
\end{array}\right)=\mathrm{tr}(\mathbf{CD}),
\end{equation}
with $$\nabla f= \left(\boldsymbol{f}_1,\boldsymbol{f}_2\right),$$
where
$$\boldsymbol{f}_1=-\frac{1}{1-\alpha}\left(\mathrm{tr}(\boldsymbol{\Sigma}_{VV}^{-1}\boldsymbol{J}^{11}),\mathrm{tr}(\boldsymbol{\Sigma}_{VV}^{-1}\boldsymbol{J}^{12}),\dots,\mathrm{tr}(\boldsymbol{\Sigma}_{VV}^{-1}\boldsymbol{J}^{pp}) \right)=-\frac{1}{1-\alpha}\mathrm{vec}'(\boldsymbol{\Sigma}_{VV}^{-1}),$$ and $$\boldsymbol{f}_4=-\alpha^{-1}\boldsymbol{f}_1=\frac{1}{\alpha(1-\alpha)}\mathrm{vec}'(\boldsymbol{\Sigma}_{VV}^{-1}).$$ It yields that 
\begin{equation*}
    \sqrt{n}\left(\mathrm{tr}(\mathbf{CD})-p-\frac{\mathrm{tr}(\boldsymbol{\Sigma}_{VV}^{-1/2}\PPi'\Z'\Z\PPi\boldsymbol{\Sigma}_{VV}^{-1/2})}{K_n}\right) \stackrel{d}{\rightarrow} N(0,\sigma_0^2),
\end{equation*}
where $\sigma_0^2=(1-\alpha)^{-2}\mathrm{vec}'(\boldsymbol{\Sigma}_{VV}^{-1})\left(\alpha^{-2}\boldsymbol{\Sigma}_2-\boldsymbol{\Sigma}_1\right)\mathrm{vec}(\boldsymbol{\Sigma}_{VV}^{-1})$.

\subsection{On the case of unbalanced instruments and non-mesokurtic errors}
We now discuss the behavior of the proposed $F_c$ test when Assumption \ref{Fassum:3} is violated, i.e., the instruments are asymptotically unbalanced and errors are non-mesokurtic. The results are summarized in the following theorem.

\begin{theorem} \label{theo:unbalanced}
    Under Assumptions \ref{assum:1} and \ref{assum:2}, as $n\rightarrow \infty$, we have
    \begin{enumerate}
        \item when ${\mathrm{E(v_1^4)}}<3\sigma_{vv}^4$, $\mathbb{P}\left( F_c>\Phi^{-1}(1-\tau) \right) \rightarrow \tau_1 \leq \tau$;
        \item when ${\mathrm{E(v_1^4)}}>3\sigma_{vv}^4$,  $\mathbb{P}\left( F_c>\Phi^{-1}(1-\tau) \right) \rightarrow \tau_2 \geq \tau$.
    \end{enumerate}
\end{theorem}
Theorem \ref{theo:unbalanced} demonstrates that the proposed $F_c$ test will exhibit size distortions with the presence of asymptotically unbalanced instruments. Furthermore, the $F_c$ test will be conservative (oversized) when the errors are platykurtic (leptokurtic). As a result, applying the limiting variance in Corallary \ref{cor:fadj} is not the perfect answer for the case of both asymptotically unbalanced instruments  and non-mesokurtic errors.

 We consider bounding $\sigma_{F}^2$ without Assumption \ref{Fassum:3}.  However, the best achievable bound is given by
     \begin{equation}
         \sigma_{F}^2 \geq \frac{2\alpha-3\omega+\alpha^2}{\alpha^2(1-\alpha)^2}:=\sigma_{L}^2 \nonumber
     \end{equation}
     so that the estimation of limiting variance can be achieved. Nonetheless, if one attempts to employ this idea to construct a corrected $F$ test, it can be demonstrated that its asymptotic size exceeds that of the proposed $F_c$ test, even when Assumption \ref{Fassum:3} is violated. Specifically, let us define
     \begin{equation}
         F_l=\hat{\sigma}_{L}^{-1}\sqrt{n}\left(F-1-\frac{C}{\sqrt{K_n}}\right), \nonumber
     \end{equation}
     then we have
     \begin{theorem} \label{theo:comp}
         Under Assumptions \ref{assum:1} and \ref{assum:2},
     \begin{equation}
         \lim_{n\rightarrow\infty}  \mathbb{P}\left( F_c>\Phi^{-1}(1-\tau) \right) \leq \lim_{n\rightarrow\infty} \mathbb{P}\left( F_l>\Phi^{-1}(1-\tau) \right) .
     \end{equation}
     \end{theorem}
     Hence, even in the presence of asymptotically unbalanced and non-mesokurtic errors, we recommend still utilizing our proposed $F_c$ test as it exhibits superior control over the size.

\noindent\textit{Proof of Theorem \ref{theo:unbalanced}}: 
Adopting the similar arguments in proving Theorem \ref{theo:F}, one can verify that 
\begin{equation}
    \mathbb{P}\left( F_c>\Phi^{-1}(1-\tau) \right) \rightarrow \Phi\left(\frac{\sigma_{F_c}}{\sigma_F}\Phi^{-1}(\tau)\right),  \; \text{with}\;\sigma_{F_c}^2=\frac{2}{\alpha(1-\alpha)}. \nonumber
\end{equation}
Therefore, the asymptotic size of the $F_c$ test depends on the variance-ratio $\sigma^2_{F_c}/\sigma^2_F$. Note that 
\begin{equation}
    \frac{\sigma^2_{F_c}}{\sigma^2_F}-1=\frac{(\omega-\alpha^2)(3-\frac{E(v_1^4)}{\sigma_{vv}^4})}{(\omega-\alpha^2)\frac{E(v_1^4)}{\sigma_{vv}^4}+2\alpha-3\omega+\alpha^2} \nonumber
\end{equation}
and $\alpha^2\leq \omega$ always hold by Cauchy-Schwarz inequality, so
 $\sigma^2_{F_c}/\sigma^2_F$ is greater (smaller) than one when the errors are platykurtic (leptokurtic). Consequently, the $F_c$ test is conservative (oversized). \\

\noindent\textit{Proof of Theorem \ref{theo:comp}}: 
Similar to the proof of Theorem \ref{theo:unbalanced}, one can show that 
\begin{equation}
    \mathbb{P}\left( F_l>\Phi^{-1}(1-\tau) \right) \rightarrow \Phi\left(\frac{\sigma_{L}}{\sigma_F}\Phi^{-1}(\tau)\right). \nonumber
\end{equation}
The proof is then completed by noticing that   $\sigma_{L}^2\leq\sigma_{F_c}^2$.
 
\end{document}